\documentclass[superscriptaddress,nofootinbib,preprintnumbers,showpacs,tightenlines,notitlepage]{revtex4}
\usepackage[latin1]{inputenc}
\usepackage{graphicx}
\usepackage{amssymb}
\usepackage{color}
\usepackage{float}
\usepackage{amsmath}
\usepackage{mathrsfs}
\usepackage{mathtools}
\usepackage{amsfonts}
\usepackage{dcolumn}
\usepackage{hyperref}
\usepackage{amsthm}
\usepackage{color}
\usepackage{upgreek}

\def\3nab{\tilde{\nabla}}

\def\be {\begin{equation}}
\def\ee {\end{equation}}
\def\ba {\begin{eqnarray}}
\def\ea {\end{eqnarray}}
\newtheorem{thm}{Theorem}

\newtheorem*{Def}{Definition}

\newcommand{\barray}{\begin{array}}
\newcommand{\earray}{\end{array}}

\newcommand{\iu}{{i\mkern1mu}}
\hypersetup{
	colorlinks=true,
	citecolor=red,
    linkcolor=blue,
    urlcolor=blue,
	}
\begin{document}

\title{Temperature of free gravitational field: A geometrical perspective}

\author{Samarjit Chakraborty}
\email{samarjitxxx@gmail.com}
\affiliation{Astrophysics Research Centre, Discipline of Mathematics, University of KwaZulu-Natal, Private Bag X54001, Durban 4000, South Africa}

\author{Gareth Amery}
\email{Ameryg1@ukzn.ac.za}
\affiliation{Astrophysics Research Centre, Discipline of Mathematics, University of KwaZulu-Natal, Private Bag X54001, Durban 4000, South Africa}

\author{Sunil D Maharaj}
\email{Maharaj@ukzn.ac.za}
\affiliation{Astrophysics Research Centre, Discipline of Mathematics, University of KwaZulu-Natal, Private Bag X54001, Durban 4000, South Africa}
 
\author{Rituparno Goswami}
\email{goswami@ukzn.ac.za}
\affiliation{Astrophysics Research Centre, Discipline of Mathematics, University of KwaZulu-Natal, Private Bag X54001, Durban 4000, South Africa}

 \begin{abstract}
 In this paper, using a novel geometrical approach, we relate the concept of the thermodynamic temperature of the free gravitational field, to the non-affinity of real null geodesics in a Newman Penrose tetrad. This naturally links various temperature functions like Clifton, Ellis and Tavakol temperature, Hawking temperature, Unruh temperature etc., in their respective proper limits. Although our analysis is done within the realm of local rotational symmetry, we show that the result can be extended to other Petrov type D geometries, like the Kerr spacetime. We also obtain the geometrical and causal transport equations for this temperature function, in the form of a hyperbolic wave equation with a forcing term, sourced by Weyl curvature and matter. Finally, as a possible physical interpretation of the non-affinity, we relate the geometrical temperature with the gravitational red/blue shift of light rays travelling along null geodesics.

 \end{abstract}

\pacs{04.20.Cv, 04.20.Dw}
\maketitle

\section{Introduction} 
 
The second law of thermodynamics encounters serious problems in scenarios where gravitational force dominates (for example in continual gravitational collapse in cosmology). In order to counter inconsistencies, Penrose proposed that the gravitational field itself should carry entropy, termed \textit{gravitational entropy} (GE) \cite{penrose1}. He suggested that this GE should be determined by the Weyl curvature tensor, and this is also known as the Weyl curvature hypothesis (WCH) \cite{penrose2}. The reason for putting forward this hypothesis is as follows. The Riemann tensor contains all the information about the spacetime curvature. The trace of this tensor (the Ricci tensor) is directly related to matter via the Einstein field equations and vanishes identically in the absence of matter. Therefore any thermodynamic description for free gravity (in the absence of matter), must be given by the tracefree part of Riemann tensor, which is also known as the Weyl curvature tensor. However, Penrose did not propose any specific functional form for this GE. \\

Significant development on this issue was reported in the work by Clifton, Ellis and Tavakol (CET) \cite{CET}, where they proposed an energy momentum tensor for the free gravitational field to be the 2-index square root of the four indexed
Bel-Robinson tensor which is constructed from the Weyl tensor and its dual ($T_{abcd}\equiv\frac14\left(C_{eabf}C^{e~~~f}_{~cd}+C^*_{eabf}C^{*e~~~f}_{~cd}\right)$). This tensor is a unique Maxwellian tensor with dimension  $L^{-4}$, so as to act as the `super energy momentum' tensor for gravitational fields \cite{E7,sen1}. Using this energy momentum tensor for free gravitational field and an ad hoc temperature function, the authors constructed the required gravitational entropy for free gravity and correctly obtained the Hawking-Bekenstein (HB) entropy for a static Schwarzschild black hole (BH) \cite{Bekenstein,Hawking} and also showed that their measure of GE remains non-negative, measures local anisotropy of the gravitational field, increases monotonically with the structure formation in the universe, and only vanishes in conformally flat spacetimes. This makes the proposition very robust as it is not limited to horizons and is applicable to any point in spacetime. \\

However, the local temperature of the free gravitational field, that the authors proposed as an inspired function from the quantum field theory in curved spacetimes and BH thermodynamics, was treated like a dimensionally consistent (with surface gravity) phenomenological integrating factor. The authors themselves argued that their proposed local gravitational temperature is an extra ingredient in their construction of the entropy of the free gravitational field. Similar construction was used by
\cite{AEGH,SG} to reproduce Hawking entropy of a black hole from continually collapsing stellar configurations. Another important insight came in \cite{bolejko}, where the proposal got validated with the help of cosmological simulation, and Bolejko demonstrated that the CET proposal is compatible with the cosmic no-hair conjecture. Meanwhile, Sussman and other collaborators conducted a series of studies on the CET proposal and provided theoretical clarity \cite{CET1,CET2,CET3,CET4,CET5,CET6,arrow,CET7}.\\

From the above examples, it is clear that the proposal for gravitational temperature in \cite{CET} has shown remarkable consistency through different systems. However, we do not have a deeper physical understanding for this proposed temperature function. Therefore our present investigation tries to answer the following question:\\

{\bf Question}: {\em Since the general relativity is purely a geometrical theory, how can we geometrically interpret the concept of the temperature of a gravitational field?}\\

To answer the above question, let us first list the basic prerequisites that any suitable candidate for the temperature function of a free gravitational field must have.
\begin{enumerate}
\item The temperature function must be a kinematic scalar, locally defined via the geometry of the spacetime on any open set of the spacetime manifold.
\item The temperature function should vanish in the absence of any intrinsic curvature (that is in Minkowski spacetime).
\item The function must reproduce all the known temperature functions (e.g. Clifton, Ellis and Tavakol temperature, Hawking temperature, Unruh temperature and so on), when applied to the respective proper limits.
\item And finally, the function must obey a causal and geometrical transport equation, to have a firm thermodynamic footing.
\end{enumerate}

In order to find a suitable candidate, that satisfies all these pre-requisites, it was observed that the CET temperature has a unique correspondence with the non-affinity of the outgoing radial null geodesics in spherically symmetric spacetimes. In other words, the very fact that the spacetime has a non-zero gravitational temperature, would make it impossible for the outgoing radial null geodesic to be affinely parametrised. Motivated by this observation, we proved an important theorem for a locally rotationally symmetric (Class II) spacetimes (of which spherical symmetry is a subclass). We showed that both the real null vectors of the Newman Penrose tetrad (which are tangent to geodesics for the given symmetry) for these spacetimes {\em cannot be simultaneously affinely parametrised}, unless the spacetime has zero intrinsic curvature (Minskowski spacetime). Without any loss of generality one can then assume that the outgoing null geodesics are not affinely parametrised and this non-affinity factor to be the temperature function. By construction, this function then obeys the first two prerequisites. In this paper, we prove that the function also obeys the other two prerequisites and hence is the ideal geometrical candidate for the gravitational temperature. \\

Once we established our results for locally rotationally symmetric (Class II) spacetimes, we could naturally extrapolate our result to more general Petrov type D spacetimes (the Kerr geometry), where one can always find a Newman Penrose tetrad with one of the real null vectors tangent to a null-geodesic. We showed that in that case as well, the non-affinity factors reproduce the temperature function on the horizon. \\

 The paper is organised as follows:. In the next section we discuss the semitetrad 1+1+2 formalism, which is best suited for LRS-II class of spacetimes, and also the Newman Penrose formalism for this class. We give a dictionary for the geometrical variables associated with these two formalisms. In section 3, we discuss the problem in the context of thermodynamics of free gravity, and in section 4 we prove the key result related to the non-affinity of real null vectors in the Newman Penrose tetrad, in presence of non-trivial intrinsic curvature. In section 5, we relate the measure of this non-affinity with the local gravitational function and show that in proper limits, this function gives all the known temperature functions. In section 6, we derive the causal transport equation for the gravitational temperature. In the next section, we showed that our final result carries on to other algebraically special spacetimes, like Kerr spacetime, beyond local rotational symmetry. And finally we give a physical interpretation of this no-affinity in terms of gravitational red/blue shift of electromagnetic waves. 

\section{Covariant description of LRS-II perfect fluid spacetimes}

The LRS II class of spacetimes (including the spherically symmetric spacetimes) are a subclass of locally rotationally symmetric (LRS) class of spacetimes, with a continuous non-trivial isotropy group of spatial rotations at every point. LRS-II spacetimes are convenient to study as they are irrotational and have a covariantly defined preferred spatial direction at each point.

\subsection{LRS-II perfect fluid spacetime in 1+1+2 semi-tetrad decomposition}

 To start with, we employ the semi-tetrad 1+1+2 covariant decomposition  for our study. For a detailed explanation please refer to \cite{H8,H9,H10,H11,H12,H13,H14}. In this scheme,
we choose a timelike unit vector $u^a$, usually defined along the fluid flow lines and the orthogonal spacelike unit vector along the preferred spatial direction $e^a$ (radial direction in spherical symmetry), and decompose the spacetime as \cite{R18}
\begin{equation}\label{decomp}
g_{ab}=-u_au_b+e_ae_b+N_{ab},
\end{equation}
where $N_{ab}$ is the projection tensor on the spherical 2-shells and is defined as $ N_{ab}=h_{ab}-e_{a}e_{b} $. Here $ h_{ab} $ is the projection tensor (in 1+3 covariant decomposition) orthogonal to $ u^{a} $ and spans the $3$-space.

Therefore, the $1+1+2$ description $\{u^{a}, e^{a}, N_{ab}\}$ gives rise to two important kind of derivatives (temporal and spatial) along $u_{a}$ and $e_{a}$: 
\begin{itemize}
\item \textbf{The dot derivative}: This is the covariant time derivative along the observers' worldlines . Therefore, for any tensor  $ S^{a...b}{}_{c...d}$, it is defined as 
$\dot{S}^{a...b}{}_{c...d}\, \equiv u^{e} \nabla_{e} {S}^{a...b}{}_{c...d} $.
\item \textbf{The hat derivative}: This is the spatial derivative along the spacelike unit vector $e^a$ defined as $\hat{\psi}_{a...b}{}^{c...d} \equiv e^{f}D_{f}\psi_{a...b}{}^{c...d}$, where 
we define
$ D_{e}S^{a...b}{}_{c...d}{} \equiv h^a{}_f
h^p{}_c...h^b{}_g h^q{}_d h^r{}_e \nabla_{r} {S}^{f...g}{}_{p...q}$.
\end{itemize}

This semi-tetrad covariant decomposition helps to obtain a set of geometrical quantities for the chosen timelike and spacelike vectors $ u^{a} $ and $ e^{a} $. These are the expansion scalar $\Theta$, acceleration $\mathcal{A}=\dot{u}^ae_a$ and the shear scalar $\Sigma=\sigma_{ab}e^ae^b$. The electric part of the Weyl tensor (responsible for tidal forces and inhomogeneity), can be similarly extracted as $\mathcal{E}=E_{ab}e^ae^b$ with $E_{ab}= C_{acbd}u^cu^d$, whereas, the magnetic part of the Weyl tensor (due to  rotation or time varying  spacetime) $H_{ab}=C^*_{acbd}u^cu^d$, vanishes in LRS-II spacetimes \cite{E7}, with $C^*_{abcd}$ as the right dual of the Weyl tensor. Moreover, these chosen unit vectors decompose the energy momentum tensor of the perfect fluid matter field (EMT) 
\begin{equation}\label{EMT}
T_{ab}=\rho u_a u_b + p \left(e_a e_b + N_{ab}\right),
\end{equation}
into the energy density $\rho$, isotropic pressure $p$, where the heat flux $ Q=0 $ and anisotropic stress $\Pi=0$. Due to spherical symmetry, the only non-vanishing geometrical quantity related to the preferred spacelike congruence $ e^{a} $ is the sheet expansion $\phi$. 
Therefore, the set of scalars that fully describe the LRS-II class of spacetimes are
\begin{equation}\label{set1}
\mathcal{D}\equiv\left\{\Theta, \mathcal{A}, \Sigma, \phi, \mathcal{E}, \rho, p \right\}.
\end{equation}
These geometrical and thermodynamical scalars (from EMT) along with their directional derivatives along $u^a$ (denoted by dot) and $e^a$ (denoted by hat) completely describe the Ricci and the Bianchi identities and therefore, can completely describe the dynamics of the spacetime. We can also geometrically define the effective `gravitational mass' as the Misner-Sharp mass $ \mathcal{M}$, enclosed within the spherical $2$-shell at a given instant of time. In terms of the aforementioned covariant scalars it is expressed as \cite{E8,E9}
\begin{equation}\label{Mass}
\mathcal{M} = \frac{1}{2K^{\frac{3}{2}}}\left(\frac{1}{3}\rho - \mathcal{E} \right).
\end{equation}
Therefore, in vacuum, the effective `gravitational mass' is entirely determined by the electric part of Weyl scalar $\mathcal{E}$. Here $K$ is the Gaussian curvature of the spherical 2-shells (which is related to the area radius of the shells $\mathcal{R}$) as the inverse square of $\mathcal{R}$. The Gaussian curvature $ K $ is related to the covariant scalars of Eq.(\ref{set1}) as  
\begin{equation}\label{gauss}
K \equiv\frac{1}{\mathcal{R}^2}=\frac{1}{3}\rho-\mathcal{E} +\frac{1}{4}\phi^2
-\left(\frac{1}{3}\Theta-\frac{1}{2}\Sigma\right)^2\, .
\end{equation}
 The field equations in terms of these scalar variables are given as\\
\textit{Propagation}:
\begin{align}
    \hat{\phi} &= -\frac{1}{2}\phi^2 + \left(\frac{1}{3}\Theta+\Sigma \right) \left(\frac{2}{3}\Theta-\Sigma \right)  -\frac{2}{3}\rho - \mathcal{E}, \label{hatphinl} \\  
    \hat{\Sigma} - \frac{2}{3}\hat{\Theta} &= -\frac{3}{2}\phi\Sigma, \label{Sigthetahat} \\  
    \hat{\mathcal{E}} - \frac{1}{3}\hat\rho &= -\frac{3}{2}\phi\mathcal{E}. \label{ehat}
\end{align}

\textit{Evolution}:
\begin{align}
    \dot\phi &= -\left(\Sigma-\frac{2}{3}\Theta \right) \left(\mathcal{A}-\frac{1}{2}\phi \right), \label{phidot} \\   
    \dot{\Sigma}-\frac{2}{3}\dot{\Theta} &= -\mathcal{A}\phi + 2\left(\frac{1}{3}\Theta-\frac{1}{2}\Sigma \right)^2  + \frac{1}{3}(\rho+3p) - \mathcal{E}, \label{Sigthetadot} \\  
    \dot{\mathcal{E}} -\frac{1}{3}\dot{\rho} &= \left(\frac{3}{2}\Sigma-\Theta \right)\mathcal{E} -\frac{1}{2}(\rho+p)\left(\Sigma-\frac{2}{3}\Theta \right). \label{edot}
\end{align}

\textit{Propagation/evolution}:
\begin{align}
    \hat{\mathcal{A}}-\dot{\Theta} &= -(\mathcal{A}+\phi)\mathcal{A} + \frac{1}{3}\Theta^2 + \frac{3}{2}\Sigma^2 + \frac{1}{2}(\rho+3p), \label{Raychaudhuri} \\
    \dot\rho &= -\Theta(\rho+p), \label{Qhat} \\    
    \hat p &= -(\rho+p)\mathcal{A}. \label{Qdot}
\end{align}

The full covariant derivative of the chosen unit timelike vector $u^a$ and the spacelike unit vector along the preferred direction $e^a$ are given by the following

\begin{align}\label{del_eu}
\nabla_au_b &=-\mathcal{A} u_ae_b+e_a e_b\left(\dfrac{1}{3}\Theta+\Sigma\right)+N_{ab}\left(\dfrac{1}{3}\Theta-\dfrac{1}{2}\Sigma\right), \nonumber\\
\nabla_a e_b &=-\mathcal{A} u_au_b +\left(\Sigma+\dfrac{1}{3}\Theta\right)e_a u_b+\dfrac{1}{2}\phi N_{ab}.
\end{align}

\subsection{LRS-II perfect fluid spacetime in Newman-Penrose tetrad formalism}
To get a transparent geometrical picture for this class of spacetimes we also use the Newman-Penrose (NP) formalism \cite{chandra,fronov}. We assume that the LRS-II spacetime is spanned by the NP null tetrad $(l^a,k^a,m^a,\bar{m}^a)$. Here  $l_a$ and $k_a$ are real null vectors, and $m_a$, $\bar{m}_a$ are the complex null vectors at a given point in the spacetime with the following properties: $l_al^a=k_ak^a=m_am^a=\bar{m}_a\bar{m}^a=0$, $l_ak^a=-1,m_a\bar{m}^a=1$ and $l_am^a=k_am^a=l_a\bar{m}^a=k_a\bar{m}^a=0$. The metric in terms of these tetrads is given by \cite{null}
\begin{equation}
g_{ab}=-2l_{(a}k_{b)}+2m_{(a}\bar{m}_{b)}.
\end{equation}

In terms of the 1+1+2 semitetrad decomposition of LRS-II spacetimes, the Newman Penrose null tetrads are  
\begin{equation}
k^{a}=\frac{A}{\sqrt{2}}(u^{a}+e^{a}),\,\,\,\, l^{a}=\frac{1}{\sqrt{2}A}(u^{a}-e^{a}),\,\,\, N_{ab}=2m_{(a}\bar{m}_{b)}.
\end{equation}
Here $A$ is a real function on the manifold, consistent with the orthonormal conditions and is associated with the Lorentz transformations (rotations of class III) that leaves the directions of the real null vectors unchanged \cite{chandra}. Therefore, we set $A=1$ without loss of generality. Our geometrical setup is done in such a manner, that the null vectors are properly aligned with the Weyl principal null directions, namely along the ingoing and outgoing radial directions in LRS-II symmetry. We can define the directional derivatives along each of the tetrad null vector fields in the following way:
\begin{equation}
D\equiv l^a\nabla_a,\; \Delta\equiv k^a\nabla_a,\; \delta\equiv m^a\nabla_a,\; \overline{\delta}\equiv \overline{m}^a\nabla_a\,.
\end{equation}
The real directional derivatives of the real null vectors are
\begin{align}
D l^{a} & =(\varepsilon + \bar{\varepsilon})l^{a}-\bar{\kappa}m^{a}-\kappa \bar{m}^{a}, \label{dirder1}\\
\Delta k^{a} & = \nu m^{a}+\bar{\nu} \bar{m}^{a}-(\gamma + \bar{\gamma})k^{a}.\label{dirder2}\\
\Delta l^{a} & =(\gamma + \bar{\gamma}) l^{a}-\bar{\tau}m^{a}-\tau\bar{m}^{a}, \label{dirder3}\\
D k^{a} & = \uppi m^{a} +\bar{\uppi}\bar{m}^{a}-(\varepsilon + \bar{\varepsilon})k^{a}.\label{dirder4}
\end{align}
The commutator relation between the $D$ and $\Delta$ directional covariant derivatives is
\begin{equation}\label{DelDcom}
\Delta D- D \Delta=(\gamma + \bar{\gamma}) D + (\varepsilon + \bar{\varepsilon}) \Delta - (\bar{\tau}+\uppi)\delta -(\tau+\bar{\uppi})\bar{\delta}.
\end{equation}
For a perfect fluid EMT \eqref{EMT}, the corresponding Ricci tensor can be expressed in NP tetrad decomposition as
\begin{equation}
R_{ab}=\frac{1}{2}(\rho+p)k_{a}k_{b}+\frac{1}{2}(\rho+p)l_{a}l_{b}+ p(k_{a}l_{b}+l_{a}k_{b})+\frac{1}{2}(\rho-p)(m_{a}\bar{m}_{b}+\bar{m}_{a}m_{b}).
\end{equation}

Therefore, we can immediately obtain the following real \textit{Ricci-NP scalars}
\begin{align}\label{NPRicci}
\Phi_{00}= & \frac{1}{2}R_{ab}l^{a}l^{b}=\frac{1}{4}(\rho+p),\nonumber \\
\Phi_{11}= & \frac{1}{4}R_{ab}(l^{a}k^{b}+m^{a}\bar{m}^{b})=\dfrac{1}{8}(\rho+p), \nonumber\\
\Phi_{22}=& \frac{1}{2}R_{ab}k^{a}k^{b}=\frac{1}{4}(\rho+p), \nonumber\\
\Lambda=& \frac{R}{24}=\frac{\rho-3p}{24}.
\end{align}
The complex Ricci-NP scalars are $\Phi_{01}=\Phi_{02}=\Phi_{12}=0$ and their complex conjugates $\left\{\Phi_{10}, \Phi_{20}, \Phi_{21}\right\}$ also vanish. 

Since, the LRS-II spacetime is Petrov type D ($\Psi_{2}\neq 0$), we have
\begin{equation}
\vert\Psi_{2}\vert=\pm \frac{\mathcal{E}}{2},
\end{equation}
and the following \textit{Weyl-NP scalars} vanish
\begin{equation}
\Psi_{0}=\Psi_{1}=\Psi_{3}=\Psi_{4}=0.
\end{equation}
We now use \eqref{del_eu} to obtain the relevant \textit{NP spin coefficients} in terms of the semi-tetrad scalars for LRS-II spacetime

\begin{align}\label{NPSpin}
&\kappa\equiv -m^{a}Dl_{a}=0,\,\,\,\tau\equiv -m^{a}\Delta l_{a}=0,\,\,\,\sigma\equiv -m^{a}\delta l_{a}=0,\,\,\, \varrho\equiv -m^{a}\bar{\delta} l_{a}=\frac{1}{2\sqrt{2}}\left(\Sigma-\frac{2\Theta}{3}+\phi\right),\nonumber\\
&\uppi \equiv \bar{m}^{a} D k_{a}=0,\,\,\,\nu \equiv \bar{m}^{a} \Delta k_{a}=0,\,\,\,\lambda \equiv \bar{m}^{a} \bar{\delta} k_{a}=0,\,\,\, \mu \equiv \bar{m}^{a} \delta k_{a}=\frac{1}{2\sqrt{2}}\left(\phi-\Sigma+\frac{2\Theta}{3}\right).
\end{align}

Here $\uppi$ denotes one of the NP spin coefficients; it should not be confused with the mathematical constant $\pi$. We now consider the following two spin coefficients
\begin{align}\label{vepsilon}
\varepsilon &\equiv -\frac{1}{2}\left(k^{a}Dl_{a}-\bar{m}^{a}D m_{a}\right),\\
\gamma &\equiv -\frac{1}{2}\left(k^{a}\Delta l_{a}-\bar{m}^{a}\Delta  m_{a}\right).\label{gamma}
\end{align}
From \eqref{dirder1} and \eqref{dirder3} we can immediately conclude that the first terms ($k^a Dl_a$ and $k^a\Delta _a$) of the above two expressions are purely real. We then use the normalisation condition $m_a\bar{m}^a=1$ and act on that with real operators $D$ and $\nabla$ to get 
\begin{align}
&\bar{m}^a Dm_a + m^a D\bar{m}_a=0=\bar{m}^a Dm_a +\overline{\bar{m}^a Dm_a},\\
&\bar{m}^a \Delta m_a + m^a \Delta\bar{m}_a=0=\bar{m}^a \Delta m_a +\overline{\bar{m}^a \Delta m_a}.
\end{align}
Thus, we conclude that in \eqref{vepsilon}--\eqref{gamma}, $\bar{m}^{a}D m_{a}$ and $\bar{m}^{a}\Delta  m_{a}$ are the purely imaginary parts of $\varepsilon$ and $\gamma$ respectively.  Hence, we get the real part of $\epsilon$ and $\gamma$ as the following
\begin{align}\label{vepsilonreal}
\Re(\varepsilon) &\equiv \dfrac{\varepsilon+\bar{\varepsilon}}{2}=-\frac{1}{2\sqrt{2}}\left(\mathcal{A}-\frac{\Theta}{3}-\Sigma\right),\\
\Re(\gamma) &\equiv \dfrac{\gamma+\bar{\gamma}}{2}=-\frac{1}{2\sqrt{2}}\left(\mathcal{A}+\frac{\Theta}{3}+\Sigma\right).\label{gammareal}
\end{align}
In the subsequent sections we only deal with the real parts of these spin coefficients and therefore the imaginary parts are not relevant for our analysis. The remaining $\alpha$ and $\beta$ could not be expressed in terms of semi-tetrad scalars, however they are not needed for our analysis. We consider the following commutator relation in NP formalism
\begin{equation}
\bar{\delta}\delta-\delta\bar{\delta}=(\bar{\mu}-\mu)D+(\bar{\varrho}-\varrho)\Delta+(\alpha-\bar{\beta})\delta-(\bar{\alpha}-\beta)\bar{\delta}.
\end{equation}
It is clear that in LRS-II spacetime, there should not be any $D$ and $\Delta$ components in the right hand side (RHS) of the relation. Therefore, $\mu$ and $\varrho$ are real quantities, which can also be seen in \eqref{NPSpin}. We also write the four equations of \textit{NP-Bianchi indentities} that are useful for our analysis as

\begin{align}
    \label{Bianchi1}
    D \Psi_{2} +\Delta \Phi_{00} + 2D\Lambda &= 3\varrho \Psi_{2} + (2\gamma + 2\bar{\gamma} - \bar{\mu})\Phi_{00}  + 2\varrho \Phi_{11}, \\[1ex]
    \label{Bianchi2}
    \Delta \Psi_{2} + D \Phi_{22} + 2\Delta \Lambda &= -3\mu \Psi_{2} + (\bar{\varrho} - 2\epsilon - 2\bar{\epsilon})\Phi_{22}  - 2\mu \Phi_{11}, \\[1ex]
     \label{Bianchi3}
    D \Phi_{11} + \Delta \Phi_{00} + 3D\Lambda &= (2\gamma - \mu + 2\bar{\gamma} - \bar{\mu})\Phi_{00}  + 2(\varrho + \bar{\varrho})\Phi_{11}, \\[1ex]
     \label{Bianchi4}
    D \Phi_{22} + \Delta \Phi_{11} + 3 \Delta \Lambda &= (\varrho + \bar{\varrho} - 2\epsilon - 2\bar{\epsilon})\Phi_{22}  - 2(\mu + \bar{\mu})\Phi_{11}.
\end{align}

\section{The context: Thermodynamics of free gravitational field}

We begin by mathematically contextualising the temperature for free gravitational fields. In order to do thermodynamics of gravitational fields, we need to have a EMT like quantity for the free fields (gravitational field without matter). According to the WCH it has to be a function of the Weyl tensor. However, Penrose did not propose any specific quantitative measure for it. Crucial insight came when Maartens and Basset \cite{E7} proved that the Bel-Robinson tensor

\small \begin{equation}\label{belrob}
T_{a b c d} \equiv \frac{1}{4} \left( C_{e a b f}C^{e \; \; \; \; f}_{\; \; c d \;
\;}+C^{*}_{\; e a b f}C^{* \; e \; \; \; \; f}_{\; \; \; \; \; c d
\; \;} \right)
\end{equation}
is a unique Maxwellian tensor that can be constructed from the Weyl tensor which acts as the \textit{super energy momentum} tensor for gravitational fields. The only issue with this tensor is that its dimension is not $L^{-2}$ (like usual EMT), rather $ L^{-4} $. To resolve this issue, Clifton, Ellis and Tavakol (CET) \cite{CET} proposed that the symmetric 2-index square root $\mathscr{T}^{grav}_{ab}$ of the Bell-Robinson tensor would act as the \textit{effective energy momentum tensor} for the free gravitational field. The authors showed that (the square root may not exist for general spacetimes) a unique square root exists (up to an arbitrary trace) for spacetimes that are Petrov type D i.e., a Coloumb like gravitational field, or Petrov type N, which describes a wave like gravitational field. In order to construct a gravitational analogue of the following thermodynamic equation 
\begin{equation}\label{second}
T_{grav} d S_{grav} =  dU_{grav} + p_{grav} dV,
\end{equation}
the authors employed the energy-momentum conservation equation for relativistic thermodynamics to obtain the following
\begin{equation}
(\rho_{grav} v)\dot{\;} +p_{grav} \dot{v} = v \Big[- u_a \nabla_{b}{\mathscr{T}_{grav}^{ab}} - \nabla_{b}q^b_{grav} - \dot{u}^a q^{grav}_{a} -\sigma_{ab} \pi_{grav}^{ab}\Big],
\end{equation}
which is like a relativistic version of \eqref{second}. Here $T_{grav}$, $S_{grav}$ and $U_{grav}$ represent the effective temperature, entropy and internal energy of the free gravitational field, respectively, and $V$ is the spatial volume. Therefore, we obtain the following fundamental thermodynamic relation
\begin{equation}
T_{grav} \dot{s}_{grav} \equiv (\rho_{grav} v)\dot{\;} +p_{grav} \dot{v}\,,
\end{equation}
where $v$ is a spatial volume element and $s_{grav}$ is the gravitational entropy density. In Petrov type D spacetimes (containing the LRS-II class) the symmetric 2 index square root of the Bell-Robinson tensor $\mathscr{T}^{grav}_{ab}$ can be written as
\begin{equation}
8\pi\mathscr{T}^{grav}_{ab}=\vert\Psi_{2}\vert\Big\{2(u_{a}u_{b}-e_{a}e_{b})+N_{ab}\Big\},
\end{equation}
with
\begin{equation}
8\pi\rho_{grav}=2\vert\Psi_{2}\vert,\,\,\,p_{grav}=0,\,\,\,8\pi\Pi_{grav}=-2\vert\Psi_{2}\vert,\,\,\,Q_{grav}=0.
\end{equation}
Here, the subscript \textit{grav} represents the gravitational counterpart (analogous) of the usual thermodynamic (matter) scalars. Therefore, due to vanishing pressure and heatflux terms, we can write the following thermodynamic identity
\begin{equation}
T_{grav}\dot{s}_{grav}=(\rho_{grav} v)\dot{\;}=-v\Big[\sigma_{ab}\pi^{ab}_{grav}+u_a \nabla_b \mathscr{T}^{ab}_{grav} \Big],
\end{equation}
where 
\begin{equation}\label{divT}
u_a \nabla_b \mathscr{T}^{ab}_{grav}=- \dot{\mathcal{E}}+\frac{3}{2}\mathcal{E}\left(\Sigma-\frac{2\Theta}{3}\right).
\end{equation}
The above quantity \eqref{divT} is proportional to $(\rho+p)\Sigma$ for a perfect fluid. Therefore, this term vanishes in vacuum, implying free gravitational field energy is conserved in absence of matter. This divergence signifies an energy exchange type of term, and in presence of matter it is not conserved and demands an extra \textit{interaction tensor} $ \Upsilon_{ab}^I $ \cite{R17}, giving rise to the process of \textit{gravitational induction}. Therefore, the energy dynamics of free gravitational fields is very rich. Now, if we want to obtain the entropy of a specific gravitational system, we can simply write
\begin{equation}
\delta s_{grav}=\dfrac{\delta(\rho_{grav}v)}{T_{grav}}
\end{equation}
and integrate over a spacelike hypersurface. If we want to proceed further and compute the above integral, we need a working formula for the gravitational temperature $T_{grav}$. In the paper \cite{CET}, the authors thus proposed the following expression for the local temperature of free gravitational fields
\begin{equation}\label{Temp}
T^{\text{CET}}_{\text{grav}} = \frac{\vert l^a k^b \nabla_{b}u_a \vert}{\pi} = \frac{\vert \dot{u}_a e^a + \Theta/3 + \sigma_{ab} e^a e^b \vert}{2 \pi}.
\end{equation}
Now that we put the gravitational temperature in a proper context, we proceed to the geometrical aspects of our null geodesics in the next section.

\section{The problem of affine parametrisation of null geodesics in LRS-II spacetime}

To start with, we quote the famous theorem by Wainwright \cite{wr}:
\begin{thm}
For all locally rotationally symmetric (LRS) spacetimes with a perfect fluid (or dust) which are not conformally flat, the Weyl tensor is type $[2, 2]$ and the two repeated principal null congruences are geodesic and shearfree.
\end{thm}
This immediately gives us the following equations of geodesics (from \eqref{dirder1}--\eqref{dirder2}):
\begin{align}\label{geodesic0}
D l^{a} & =(\varepsilon + \bar{\varepsilon})l^{a}=2\Re(\varepsilon) l^{a}, \nonumber\\
\Delta k^{a} & =-(\gamma + \bar{\gamma})k^{a}= - 2\Re(\gamma)  k^{a}.
\end{align}
Also the shearfree geodesic condition provides us with the following vanishing spin coefficients
\begin{equation}\label{cond1}
\kappa=\sigma=0, \nu=\lambda=0.
\end{equation}
Now we ask the following question: Can we affinely parametrize both (ingoing and outgoing) the null geodesics simultaneously without any loss of generality? To answer this question, we now state and prove a very important theorem for perfect fluid spacetimes in LRS-II symmetry.

\begin{thm}
In a LRS-II spacetime, both ingoing and outgoing radial null geodesics can be parametrised affinely, if and only if the spacetime is Minkowski.
\end{thm}

\begin{proof}
We prove this theorem by contradiction. Therefore, let us assume that for any given spacetime (either vacuum or with a perfect fluid obeying all energy conditions) both the ingoing and outgoing null geodesics are affinely parametrized. Hence, from \eqref{geodesic0} we obtain 
\begin{equation}\label{cond2}
\Re(\varepsilon)=\Re(\gamma)=0.
\end{equation}
Next, we employ the commutator relation between the $D$ and $\Delta$ directional covariant derivatives \eqref{DelDcom}. 
As we are considering LRS II spacetimes, $\delta$ and $\bar{\delta}$ components should  vanish in the commutator relation due to rotational symmetry. Therefore, LRS II symmetry demands 
\begin{equation}\label{cond3}
(\bar{\tau}+\uppi)=(\tau+\bar{\uppi})=0.
\end{equation}
Similarly, we employ the directional derivatives of the real null vectors along each others direction \eqref{dirder3}--\eqref{dirder4}. Since, we are in LRS II spacetimes, the components along the complex null vectors should vanish (on the 2 shell) and therefore, due to the spacetime symmetry we get  
\begin{equation}\label{cond4}
\tau=\uppi=0.
\end{equation}
We now consider the following NP field equation
\begin{equation}\label{NPf}
D\gamma-\Delta\varepsilon =  \, (\tau+\bar{\uppi})\alpha+(\bar{\tau}+\uppi)\beta-2\Re(\varepsilon)\gamma-2\Re(\gamma)\varepsilon   +\tau\uppi - \nu \kappa  +\Psi_{2}+\Phi_{11}-\Lambda.
\end{equation}
Therefore, in LRS II spcaetime when both the ingoing and outgoing null geodesics are affinely parametrized, we substitute \eqref{cond1}--\eqref{cond4} in the above equation \eqref{NPf} to obtain 
\begin{equation}\label{precons}
D\gamma-\Delta\varepsilon =\Psi_{2}+\Phi_{11}-\Lambda.
\end{equation}
Since $l^a$ and $k^a$ are real null vectors and we are in LRS II symmetry, the spin coefficients $\epsilon$ and $\gamma$ are real quantities (as can be seen from \eqref{geodesic0}) and therefore, using \eqref{cond2} we also have 
\begin{equation}\label{cond5}
\epsilon = \gamma = 0.
\end{equation}
Hence, for a perfect fluid matter field (using \eqref{cond5} in \eqref{precons}) we have the following constraint equation
\begin{align}\label{cons0}
\Psi_{2}+\Phi_{11}-\Lambda=0.
\end{align}
We also demand that the above condition $ \mathcal{C}\equiv \Psi_{2}+\Phi_{11}-\Lambda=0  $ is to be satisfied during the evolution and propagation of the system. Therefore, we demand 
\begin{equation}\label{DdeltaC}
D\mathcal{C}=0 ,\,\,\,\ \Delta \mathcal{C}=0.
\end{equation}
Using the Bianchi identities \eqref{Bianchi1}--\eqref{Bianchi4} the respective conditions \eqref{DdeltaC} become the following
\begin{align}
\Delta \Phi_{00}+3D\Lambda + \frac{3}{2}(\mu \Phi_{00}-\varrho \Psi_{2})-3\varrho \Phi_{11}=0, \label{DC}\\
D \Phi_{22}+3 \Delta \Lambda + \frac{3}{2}(\mu \Psi_{2}-\varrho \Phi_{22})+3\mu \Phi_{11}=0. \label{DelC}
\end{align}
By substituting $\Psi_{2}$ from \eqref{DC} in \eqref{DelC}, we obtain the following condition that constrains the matter field
\begin{equation}\label{cons1}
\frac{3}{2}(\mu^2 - \varrho^2)\tilde{\Phi}+3(\mu D\Lambda + \varrho \Delta \Lambda)+ \mu \Delta \tilde{\Phi}+\varrho D\tilde{\Phi}=0,
\end{equation}
where the corresponding NP-Ricci scalars are given by \eqref{NPRicci} with $\tilde{\Phi}\equiv\Phi_{00}=\Phi_{22}$.

For the transparency of understanding, we now consider the four cases separately:
\begin{itemize}
\item[(a)] \textbf{Vacuum:} For the trivial vacuum spacetime we have $ \rho=p=0 $ and therefore, the constraint Eq\eqref{cons0} implies a conformally flat spacetime, i.e.,
\begin{equation}
\Psi_{2}=0.
\end{equation}
This is only possible if the spacetime is Minkowski. Therefore, the vacuum Minkowski spacetime is a solution of Eq\eqref{cons0}. 

\item[(b)] \textbf{$\mathbf{\Lambda}$-Vacuum:} It is to be noted that $ (\rho+p)=0 $ is also possible in a $ \mathbf{\Lambda}$-vacuum, i.e., de Sitter (dS) and Anti-de Sitter (AdS) spacetimes, with $ \mathbf{\Lambda}$ as the cosmological constant and $\rho=-p=\mathbf{\Lambda}$. However, again from the original constraint equation \eqref{cons0} it becomes clear that
\begin{equation}
\Psi_{2} =\Lambda=\frac{\mathbf{\Lambda}}{6},\,\,\, \text{with}\,\,\, \Phi_{11}=0.
\end{equation}
Now, we know that both de Sitter (dS) and Anti-de Sitter (AdS) spacetimes are conformally flat and therefore $\Psi_{2} =0$. Hence for the constraint \eqref{cons0} to be satisfied, the only possibility is the trivial ($ \mathbf{\Lambda}=0 $) Minkowski spacetime.
 
\item[(c)]\textbf{Dust:} Now we check whether any dust solution ($p=0$) is compatible with Eq\eqref{cons0}, i.e.,
\begin{equation}\label{cons00}
\mathcal{E}=\pm \frac{\rho}{6}. 
\end{equation}
Using the evolution and propagation equations for dust matter field we get $\mathcal{A}=0$ and $\dot{\rho}=-\Theta \rho$. We employ the Weyl evolution equation \eqref{edot}
and demand that Eq\eqref{cons0} holds true for all time. Consequently, we obtain that the spacetime has to be shear-free i.e., 
\begin{equation}
\Sigma=0.
\end{equation}
Additionally we use the Raychaudhuri equation \eqref{Raychaudhuri} and write the equation of shear evolution \eqref{Sigthetadot} as 
\begin{equation}
\dot{\Sigma}=-\frac{\Sigma^2}{2}-\frac{2}{3}\Theta \Sigma - \mathcal{E}.
\end{equation}
Now if we demand that $ \Sigma=0 $ should be true for all time ($\dot{\Sigma}=0$), then the spacetime has to be conformally flat, i.e.,
\begin{equation}
 \mathcal{E}=0.
\end{equation}
Moreover, from the Weyl evolution equation \eqref{edot} we obtain $\dot{\mathcal{E}}=-\Theta \mathcal{E}$, therefore the spacetime remains conformally flat for all time. However, we started with the condition \eqref{cons00}. Therefore, no standard dust solution is viable and the only possible solution is 
\begin{equation}
\rho=0,
\end{equation}
 i.e., a flat Minkowski vacuum. This could also be transparently argued that for $\Sigma=0$, the Le\-ma\^itre--Tolman--Bondi (LTB) spacetime is necessarily dust Friedmann--Le\-ma\^itre--Robertson--Walker (FLRW) metric, which is by definition Weyl free \cite{dust}, and under the constraint \eqref{cons00} becomes Minkowski.

\item[(d)]\textbf{General perfect fluid:} From \eqref{cons1} we obtain that when $ (\rho+p)\neq 0 $ (not in vacuum or $\mathbf{\Lambda}$-vacuum), the speed of sound ($c_{s}^2 \neq 0$ i.e., not dust) for a general perfect fluid has to obey the following relation
{\begin{equation}
c_{s}^2=\dfrac{{3\phi}\left(\Sigma-\frac{2\Theta}{3}\right)+\left({3\phi \Theta}-\frac{\mathcal{A}}{c_{s}^2}\left(\Sigma-\frac{2\Theta}{3}\right)   \right)}{{\phi\Theta}+\frac{5\mathcal{A}}{c_{s}^2}\left(\Sigma-\frac{2\Theta}{3}\right)}.
\end{equation}
}Here the covariant directional derivatives of $p$ are related to that of the energy density $ \rho $ through the speed of sound as
\begin{equation}
\dot{p}=c_{s}^{2}\dot{\rho},\,\,\,\hat{p}=c_{s}^{2}\hat{\rho},\,\,\, c_{s}^2\equiv\partial p/\partial \rho.
\end{equation}
Now, we can always fix our frame so that $ \dot{K}=0 $, i.e., $ \Sigma-\frac{2\Theta}{3}=0 $. Therefore the above relation yields
\begin{equation}
 c_{s}^2=3,
\end{equation} 
which is unphysical as it violates the dominant energy condition (DEC).
\end{itemize}
Hence, for all cases it is proved that unless the spacetime is Minkowski, both the null geodesics cannot be affinely parametrised.
\end{proof}

Therefore, in LRS II spacetimes with perfect fluid, at least one of the null geodesics has to be non-affinely parametrised without any loss of generality.
Hence, we take the outgoing null geodesic to be non-affinely parametrised   and the resulting geodesic equations become 
\begin{equation}\label{nonaff}
\Delta k^{a} = -2\Re(\gamma)k^{a}.
\end{equation}
Let us call the factor $-2\Re(\gamma)\equiv \mathcal{K} $
as the `non-affinity factor' of an outgoing null geodesic, signifying its degree of non-affinity.

\section{Temperature of gravitational field as non-affinity factor}
From the arguments above, we now propose that this non-affinity factor can be identified as the temperature of the free gravitational field. Hence, we define the local temperature of the free gravitational field as the following
\begin{Def}
In any open set $\mathscr{U}$ of manifold $\mathscr{M}$, the local temperature function of the free gravitational field is proportional to the non-affinity factor of the outgoing radial null geodesics of the Newman-Penrose tetrad in LRS-II spacetimes. In other words,  the real part of the spin coefficient $\gamma$ determines the local gravitational temperature, i.e., 
\begin{equation} \label{T_new}
T_{grav}=\Big\vert\frac{\mathcal{K}}{\sqrt{2}\pi}\Big\vert = \left\vert \frac{1}{\sqrt{2}\pi} (\gamma + \bar{\gamma}) \right\vert.
\end{equation}
\end{Def}
The $\sqrt{2}$ in the denominator is due to the normalisation of null vectors, for clarity and consistency we have kept our vector convention similar to the CET paper \cite{CET}. Now, we proceed to show that the above definition is consistent with other propositions of gravitational temperature in their appropriate limits. We further analyse the proposed local temperature of the free gravitational field in the context of two separate cases to establish our point.
\subsection{Local CET temperature of the free gravitational field}
In the CET paper \cite{CET}, the temperature was defined as the projection of the covariant derivative of unit timelike vector $u^a$, onto the plane spanned by the ingoing and outgoing null vectors $k^a, l^a$ \eqref{Temp}. Using the $\Delta$ directional derivative of the null vectors \eqref{dirder2}--\eqref{dirder3}, Surprisingly, we find that the $ T^{\text{CET}}_{\text{grav}} $ is directly proportional to the non-affinity factor of the outgoing null geodesics and is identical to our proposed definition \eqref{T_new}. 
Now that we have validated our definition of $ T_{grav} $ in NP formalism, we proceed to test it in the 1+1+2 semi-tetrad formalism. We start with a congruence of null geodesics \eqref{nonaff} which is not affinely-parametrized with the equation 
\begin{equation}
\Delta k_{b}\equiv k^{a}\nabla_{a}k_{b}=\mathcal{K} k_{b},\Rightarrow \mathcal{K}=-l^{b}k^{a}\nabla_{a}k_{b},
\end{equation}
where the quantity $\mathcal{K}$ measures the failure of $k^{a}$ to be affinely-parametrized. Now, in terms of our semi-tetrad basis vectors $\{u^a, e^a\}$, we express our null vectors $\{k^a, l^a\}$. Subsequently, by using the full covariant derivative of $e_a$ and $u_a$ in \eqref{del_eu}, we obtain the `non-affinity factor' in semi-tetrad decomposition as
\begin{equation}
\mathcal{K} =\frac{1}{\sqrt{2}}\left(\mathcal{A}+\frac{\Theta}{3}+\Sigma\right).
\end{equation}
Since, $ \mathcal{K}=-(\gamma+\bar{\gamma}) $ (also notice \eqref{gammareal}), we finally obtain the following
\begin{align}
T^{\text{CET}}_{\text{grav}} &\equiv \frac{\mid l^{a}k^{b}\nabla_{b}u_{a}\mid}{\pi}\nonumber\\
& =\left\vert \frac{1}{\sqrt{2}\pi} (\gamma + \bar{\gamma})    \right\vert\nonumber\\
&\equiv T_{grav}\nonumber\\
& =\dfrac{\left\vert \mathcal{A}+\frac{\Theta}{3}+\Sigma\right\vert}{2\pi}.
\end{align}
This makes our approach consistent in both the paradigms. Therefore, we establish that the local $T_{grav}$ proposed by CET originates from the local non-affinity factor of the null geodesics. In the NP formalism, the spin coefficient $ \gamma $ is a complex scalar that encodes the relative acceleration and rotation between the null tetrad vectors. It's real part measures the gravitational redshift or acceleration experienced along the null rays, while its imaginary part captures frame rotation relative to the transverse spatial directions. In LRS-II spacetimes, $\gamma$ can typically be chosen real and evolves along outgoing null geodesics according to the equation
\begin{equation}
D \gamma=\Psi_{2}+\Phi_{11}-\Lambda,
\end{equation}
directly linking its variation to the coulomb Weyl curvature (Petrov type D) and Ricci curvature. Thus, $\gamma$ acts as a gravitational potential variable that characterises the local gravitational field strength and its evolution along light rays.

\subsection{Black hole horizon temperature}
In LRS-II spacetimes, the symmetry forces the following spin coefficients to vanish identically
\begin{equation}
\kappa = \tau = \sigma = \uppi = \nu = \lambda = 0,
\end{equation}
 with only $\varrho$ and $\mu$ (both real) non-zero. The Frobenius theorem guarantees hypersurface orthogonality, eliminating twist (trivial in LRS-II). In our convention, the outgoing null vector is $k^a$, associated with $\gamma + \bar{\gamma}$, and its 
expansion is governed by $\mu$. Hence, a marginally outer trapped surface is defined by $\mu = 0$. The NP expansion for $k^a$ simplifies everywhere to $k^a \nabla_a k^b = -(\gamma + \bar{\gamma}) k^b$ because $\nu = 0$ in LRS-II. Thus $k^a$ is everywhere geodesic, and $-(\gamma + \bar{\gamma})$ is its non-affinity parameter. The surface gravity $\mathcal{K}_{\mathcal{H}}$ is defined by $k^a \nabla_a k^b = \mathcal{K}_{\mathcal{H}} k^b$ on the horizon ${\mathcal{H}}$ (with appropriately normalised null vector $k^a$), yielding 
the universal relation 
\begin{equation}
\mathcal{K}_{\mathcal{H}} = -(\gamma + \bar{\gamma}).
\end{equation}

In Nielsen's thesis \cite{nielsen}, where the outgoing null vector is $l^a$ (opposite to our vector convention) and surface gravity is $\varepsilon + \bar{\varepsilon}$, he explicitly notes that for a geodesic tangent to $k^a$ (in our convention outgoing) the surface gravity becomes $-(\gamma + \bar{\gamma})$. For Killing horizons, the Killing symmetry together with the Frobenius theorem forces the transverse terms to vanish, making $-(\gamma + \bar{\gamma})$ the 
unique surface gravity. For isolated horizons, the non-expanding condition $\mu = 0$ and LRS-II identities guarantee $\nu = 0$, yielding $\mathcal{K}_{\mathcal{H}} = -(\gamma + \bar{\gamma})$ as an intrinsic invariant. For local horizons, defined by $\mu = 0$, the same relation holds identically, but its interpretation differs: the value depends on the choice of normalization of $k^a$, with different dynamical prescriptions reducing to $-(\gamma + \bar{\gamma})$ evaluated on the horizon. Thus, across all three horizon types, Nielsen's analysis confirms that in LRS-II symmetry, the surface gravity is $\mathcal{K}_{\mathcal{H}} = -(\gamma + \bar{\gamma})$. Therefore, $T_{grav}$ also corresponds to black hole horizon temperature in appropriate limits. To illustrate the point further, we consider the Schwarzschild spacetime, where 
\begin{equation}
\mathcal{A}=\frac{\mathcal{M}}{r^2}\left(1-\frac{2\mathcal{M}}{r}\right)^{-1/2},\,\,\, \Theta=\Sigma=0.
\end{equation}
Therefore, the temperature of free gravitational field becomes
\begin{equation}
T_{grav}=\frac{\mathcal{A}}{2\pi}.
\end{equation}
The temperature measured by a local static observer diverges at the horizon because their proper acceleration becomes infinite, requiring the redshift factor to yield the finite Hawking-Bekenstein (HB) temperature observed at infinity
\begin{equation}
T_{HB}=\frac{\mathcal{A}}{2\pi}\sqrt{1-\frac{2\mathcal{M}}{r}}\Bigg\vert_{r=2\mathcal{M}}=\frac{1}{8\pi\mathcal{M}}.
\end{equation}

This can also be directly observed when we employ the Kinnersley null tetrads (naturally adapted to infinity)  with the normalisation conditions $k^al_a=-1$ 
\begin{align}
k^{a}=& \left(\frac{1}{1-\frac{2\mathcal{M}}{r}}, 1, 0, 0\right),\nonumber\\
l^{a}=& \left(\frac{1}{2}, -\frac{1}{2}\left(1-\frac{2\mathcal{M}}{r}\right), 0, 0 \right).
\end{align}
In this frame the spin coefficient $\gamma$ becomes $\gamma={\mathcal{M}}/{2r^2}$, which directly gives us the standard expression of the surface gravity of Schwarzschild BH
\begin{equation}\label{Kschw}
\kappa_{\mathcal{H}}=\vert \gamma+ \bar{\gamma}\vert_{\mathcal{H}}=\frac{1}{4\mathcal{M}}.
\end{equation} 

Therefore, the local gravitational temperature from the non-affinity factor of null geodesic corresponds to both the CET and BH temperature, creating a consistent picture. 

\section{Hyperbolic transport equation of the gravitational temperature}

Since we are building the concept of temperature via a  geometric perspective, causality is inbuilt in our mechanism. Therefore, the transport equation for free gravitational field is necessarily hyperbolic in nature, contrary to classical fluid mechanics, where it is elliptical in nature.
Using the real part of the \eqref{NPf} we obtain the following
\begin{equation}\label{transport1}
\frac{\pi}{\sqrt{2}}D T_{grav}=\Delta \Re(\epsilon) -2\sqrt{2}\pi \Re(\epsilon) T_{grav}+\Re(\Psi_{2}) + \Phi_{11}-\Lambda,
\end{equation}
where, using \eqref{vepsilonreal}, we can always write the following 
\begin{equation}\label{Deltaep}
\Delta \Re(\epsilon)=-\frac{1}{2\sqrt{2}} \Delta \left(\mathcal{A}-\frac{\Theta}{3}-\Sigma \right)=-\frac{1}{2\sqrt{2}} \left[\Delta \left(\mathcal{A}-\Theta\right)-\Delta \left(\Sigma -\frac{2\Theta}{3}\right) \right].
\end{equation}
In LRS II sapcetime, for any scalar ${f}$ we can always write the following commutation relation 
\begin{equation}\label{commute}
\hat{\dot{{f}}}-\dot{\hat{{f}}}=-\mathcal{A}\dot{f}+\left(\frac{\Theta}{3}+\Sigma\right)\hat{f}.
\end{equation}
We also express the real NP derivatives in terms of the semi-tetrad directional derivatives as
\begin{equation}
\Delta {f}=\frac{1}{\sqrt{2}}(\dot{{f}}+\hat{{f}}),\,\,\, D{f}=\frac{1}{\sqrt{2}}(\dot{{f}}-\hat{{f}}).
\end{equation}
Consequently, using the field equations \eqref{Sigthetahat}, \eqref{Sigthetadot} and \eqref{Raychaudhuri} we obtain the following $\Delta$ derivatives in terms of the 1+1+2 scalars
\begin{align}
\Delta \left(\Sigma- \frac{2\Theta}{3}\right)=&\frac{1}{\sqrt{2}}\left[-\mathcal{A}\phi-\frac{3}{2}\phi \Sigma + \frac{1}{2}\left(\Sigma-\frac{2\Theta}{3}\right)^2+\frac{1}{3}(\rho+3p)-\mathcal{E}\right],\label{Delstatic}\\
\Delta \left(\mathcal{A}-\Theta\right)=&\frac{1}{\sqrt{2}}\left[-(\mathcal{A}+\phi)\mathcal{A}+\frac{\Theta^2}{3}+\frac{3}{2}\Sigma^2+\frac{1}{2}(\rho+3p)+(\dot{A}-\hat{\Theta})\right].
\label{Delatheta}
\end{align}
Here, $\dot{\mathcal{A}}$ and $\hat{\Theta}$ are related to each other by the following relation (for non-dust matter field)
\begin{equation}
\dot{\mathcal{A}}=-\dfrac{1}{c^2_s}(\rho+p)\dfrac{\partial^2 p}{\partial \rho^2} \Theta \mathcal{A}+ c^2_{s} \hat{\Theta}.
\end{equation}
The above relation can be obtained using \eqref{commute} for $p$ along with the field equations. Therefore, using \eqref{Deltaep}, \eqref{Delstatic} and \eqref{Delatheta}, we express the NP $\Delta$ derivative of $\epsilon$ in terms of the 1+1+2 scalars
\begin{equation}\label{Deltaepf}
\Delta \Re(\epsilon)=2\Re(\epsilon)\Re(\gamma)-\frac{1}{4}\left[\frac{1}{6}(\rho+3p)+\dot{\mathcal{A}}-\hat{\Theta}+\frac{3}{2}\phi\Sigma +\mathcal{E}\right].
\end{equation}
Consequently, from \eqref{transport1}, we obtain the following general transport equation of the temperature of the free gravitational field
\begin{equation}\label{transport2}
\dot{T}_{grav}-\hat{T}_{grav}= \left(\mathcal{A}-\frac{\Theta}{3}-\Sigma\right)T_{grav}-\dfrac{1}{2\pi}\left[\frac{1}{6}(\rho+3p)+\dot{\mathcal{A}}-\hat{\Theta}+\frac{3}{2}\phi\Sigma +\mathcal{E}\right]+ \frac{{1}}{\pi} \left[ {\mathcal{E}}+ \dfrac{\rho+3p}{6}\right].
\end{equation}

Since, in LRS II spacetimes the magnetic part of the Weyl tensor vanishes and our geodesics are aligned with the principal null directions, we can take $ \Psi_{2} $ to be real. It is important to note that, \eqref{transport2} is true when both the outgoing and ingoing null geodesics are not affinely parametrised i.e., $\Re(\gamma)$ and $\Re(\epsilon)$ are non-zero.

Now, if we affinely parametrise the ingoing null geodesic i.e., $\Re(\epsilon)=0$ and demand that it remains so during the propagation and evolution of the system, i.e., $\Delta \Re(\epsilon)=0$, (from \eqref{Deltaepf}), then we obtain the following simplified transport equation of the gravitational temperature 
\begin{equation}\label{DT_g}
D T_{grav}=\frac{\sqrt{2}}{\pi}\left[\Psi_{2}+\Phi_{11}-\Lambda\right].
\end{equation}
Now that we have a simple working equation for the dynamics of $T_{grav}$, we can observe some interesting features. Interestingly, we can see how the directional derivative of the gravitational temperature is governed by both the matter and Weyl sector. Physically it makes sense as the temperature depends on the kinematical scalars. In vacuum $\Phi_{11}=\Lambda=0$ we have 
\begin{equation}
DT_{grav}^{vac}=\frac{\sqrt{2}}{\pi}\Psi_{2},
\end{equation}
making the evolution and propagation of the gravitational temperature entirely dependent on the Weyl scalar, as it is purely driven by the free gravitational field. In a conformally flat spacetime ($\Psi_{2}=0$) the variation of temperature  becomes
\begin{equation}
DT_{grav}^{cf}=\frac{\sqrt{2}}{\pi}(\Phi_{11}-\Lambda).
\end{equation}
Interestingly, in this case the free gravitational energy density vanishes (proportional to Weyl scalar) but the temperature and its variation remains finite.
Subsequently, by taking the $\Delta$ derivative of the transport equation \eqref{DT_g} and using the Bianchi identities in \eqref{Bianchi1}--\eqref{Bianchi4}, we arrive at the following second order wave equation 
\begin{equation}\label{transfer}
\Delta D T_{grav} =\frac{\sqrt{2}}{\pi}\Big[-2D\Phi_{22}-6\Delta \Lambda-3\mu(\Psi_{2}+2\Phi_{11})+3\varrho \Phi_{22} \Big],
\end{equation}
where the NP derivatives of the Ricci scalars can be expressed in terms of the 1+1+2 covariant scalars in the following manner
\begin{equation}\label{compo1}
D \Phi_{22}=\dfrac{(\rho + p)(1+c_{s}^2)}{4\sqrt{2}}\left(\frac{\mathcal{A}}{c_{s}^2}-\Theta \right),
\end{equation}
and 
\begin{equation}\label{compo2}
\Delta \Lambda= \dfrac{(\rho + p)(3c_{s}^2-1)}{24\sqrt{2}}\left(\Theta + \frac{\mathcal{A}}{c_{s}^2}\right).
\end{equation}
To obtain the expressions \eqref{compo1} and \eqref{compo2} we have used the evolution and propagation equations for perfect fluid \eqref{Qhat}--\eqref{Qdot}. This clearly, demonstrates how the gravitational temperature wave equation depends on the inhomogeneity of matter energy density (through Weyl) and pressure. Additionally, for any scalar $f$ we can write the following 
\begin{equation}
\Delta D f=\frac{1}{2}\left(\ddot{f}-\dot{\hat{f}}+\hat{\dot{f}}-\hat{\hat{f}}\right).
\end{equation}
Therefore, using \eqref{commute}, the LHS of \eqref{transfer} becomes
\begin{equation}
\Delta D T_{grav} =\frac{1}{2}\Big[\ddot{T}_{grav}-\hat{\hat{T}}_{grav}-\mathcal{A}\dot{T}_{grav}+\left(\Sigma+\frac{\Theta}{3}\right)\hat{T}_{grav}\Big],
\end{equation}
and on the RHS of \eqref{transfer}, the Ricci-NP scalars are given as \eqref{NPRicci} and the NP-spin coefficients as \eqref{NPSpin}. Therefore, using $\eqref{DT_g}$, we can now obtain the wave equation of the gravitational temperature in the 1+1+2 semi-tetrad formalism, clearly showing its dependence on the geometric and thermodynamic scalars (with $\mathcal{A}=\Sigma + \Theta/3$, i.e., $\Re(\epsilon)=0$). For a general non-dust matter field the wave equation of the gravitational temperature becomes
\begin{widetext}
\begin{equation}\label{wave}
 \ddot{T}_{grav}-\hat{\hat{T}}_{grav}=\frac{1}{\pi}\left(\mathcal{E}+\dfrac{\rho+3p}{6}\right)\left(\Sigma+\frac{\Theta}{3}\right)-\frac{3}{2\pi}\left(\phi-\Sigma+\frac{2\Theta}{3}\right)\mathcal{E}-\frac{(\rho+p)}{2\pi}\left[\left(\Sigma+\frac{\Theta}{3}\right)\left(5+\frac{1}{c_{s}^2}\right)+\Theta \left(c_{s}^2 -1\right)-3\Sigma\right].
 \end{equation}
 Similarly, for a dust matter field (with $\mathcal{A}=\Sigma + \Theta/3=0$), the second order wave equation is obtained as the following
 \begin{equation}
 \ddot{T}_{grav}-\hat{\hat{T}}_{grav}=-\frac{3}{2\pi}\left(\phi + \Theta \right)\mathcal{E}+\frac{\hat{\rho}}{2\pi}.\label{wave_dust}
\end{equation}
\end{widetext}
It is now straightforward to observe that the terms involving the matter field energy density and pressure in \eqref{wave} and \eqref{wave_dust} vanish in vacuum. 

As shown in \cite{MOTS1}, It is interesting to notice that on the marginally outer trapped surface (MOTS), the second term completely vanishes in \eqref{wave}, as the expansion of the outgoing null geodesic vanishes i.e.,
\begin{equation}\label{MOTS}
\mu=0 \Rightarrow \left(\phi-\Sigma+\frac{2\Theta}{3}\right)=0. 
\end{equation}
From the above equations \eqref{wave} and \eqref{MOTS}, we can immediately see that the wave equation (for non-dust matter field) of the gravitational temperature on the dynamical horizon (horizon of continuously accreting BH) becomes
\begin{equation}
\ddot{T}_{grav}-\hat{\hat{T}}_{grav}\overset{\text{MOTS}}{=}\frac{1}{\pi}\left(\mathcal{E}+\dfrac{\rho+3p}{6}\right)\left(\Sigma+\frac{\Theta}{3}\right)-\frac{(\rho+p)}{2\pi}\left[\left(\Sigma+\frac{\Theta}{3}\right)\left(5+\frac{1}{c_{s}^2}\right)+\Theta \left(c_{s}^2 -1\right)-3\Sigma\right],
\end{equation}
where $\overset{\text{MOTS}}{=}$ denotes that the equation holds on the MOTS. A similar equation for dynamical horizon follows from \eqref{wave_dust} and \eqref{MOTS} in presence of dust matter field 
\begin{equation}
\ddot{T}_{grav}-\hat{\hat{T}}_{grav}\overset{\text{MOTS}}{=}\frac{\hat{\rho}}{2\pi}. 
\end{equation}
Therefore, in the dust case, the second order wave equation of $T_{grav}$ has no explicit dependence on the Weyl scalar. Rather it is driven by the density inhomogeneity of the dust field. 

\section{Beyond LRS-II symmetry: The Kerr Spacetime}

In the previous sections we dealt with LRS-II symmetry and proved a theorem that conclusively showed the viability of our proposal. We now consider spacetimes that are beyond this symmetry to check whether $T_{grav}$ remains the same. To do our analysis we consider a general Petrov type D vacuum spacetime where we have 
\begin{equation}
\left\{\mathcal{A},\Omega, \Sigma, \mathcal{E}, \mathcal{H}, \phi, \xi \right\},
\end{equation}
as the irreducible set of covariant scalars and the corresponding set of vectors and tensors is 
\begin{equation}
\left\{\mathcal{A}_a, \Omega_a, \Sigma_a, \alpha_a, a_a, \mathcal{E}_a, \mathcal{H}_a, \Sigma_{ab}, \zeta_{ab}, \mathcal{E}_{ab}, \mathcal{H}_{ab}\right\}.
\end{equation}
Apart from the scalars previously described for LRS-II, we now have a set of other scalars, vectors and tensors (please refer to \cite{CH} for their usual meaning). Let us now consider a congruence of null geodesics that are not affinely parametrised with $k^a$ and $l^a$ as the tangent vectors. We can always split them in terms of the chosen unit timelike vector $u^a$ and the orthogonal unit spacelike vector along a chosen preferred direction $e^a$. By considering the null geodesic equation we obtain the non-affinity factor as 
\begin{equation} \label{kappa_gen}
\mathcal{K}=-\frac{1}{2\sqrt{2}}(u^b-e^b)(u^a+e^a)\nabla_{a}(u_b+e_b),
\end{equation}
where the quantity $\mathcal{K}$ measures the failure of $k^{a}$ to be affinely-parametrized. Now, in order to express $\mathcal{K}$ in terms of our new set of irreducible covariant quantities we need to consider the full covariant derivative of our chosen unit vectors 
\begin{align}\label{del_u_gen}
\nabla_au_b=& -u_a(\mathcal{A}e_b +\mathcal{A}_b)
+e_a e_b\left(\dfrac{1}{3}\Theta+\Sigma\right)
+e_a(\Sigma_b + \epsilon_{bc}\Omega^c)+(\Sigma_a - \epsilon_{ac}\Omega^c)e_b +N_{ab}\left(\dfrac{1}{3}\Theta-\dfrac{1}{2}\Sigma\right)+\Omega \epsilon_{ab}+\Sigma_{ab},\\
\nabla_a e_b=& -\mathcal{A} u_au_b-u_a \alpha_b
 +\left(\Sigma+\dfrac{1}{3}\Theta\right)e_a u_b + \left(\Sigma_a - \epsilon_{ac}\Omega^c \right)u_b +e_a a_b +\dfrac{1}{2}\phi N_{ab}+ \xi \epsilon_{ab}+\zeta_{ab}.
\label{del_e_gen}
\end{align}
By substituting \eqref{del_u_gen} and \eqref{del_e_gen} in \eqref{kappa_gen}, we obtain the non-affinity factor of our outgoing null geodesic as
\begin{equation}\label{kappa_semit}
\mathcal{K}=\frac{1}{\sqrt{2}} \left(\mathcal{A}+\frac{\Theta}{3}+\Sigma\right).
\end{equation}
Surprisingly, we again obtain the expression that corresponds to the CET temperature. In terms of the NP formalism we can write \eqref{kappa_gen} and \eqref{kappa_semit} as
\begin{equation}
\mathcal{K}=-l^b\Delta k_b=-(\gamma+\bar{\gamma}).
\end{equation}
This is the same expression we obtained previously in \eqref{nonaff}. Now, keeping in mind of our study in LRS-II, we can expect that this quantity gives rise to the $T_{grav}$ in Petrov type D spacetime, although a rigorous proof is required, which we leave for future work. However, we will try to motivate the argument by considering a specific example, e.g., the Kerr spacetime. 

In Boyer-Lindquist coordinates the Kerr metric is
\begin{equation}
ds^2= -\left(1-\frac{2Mr}{\mathscr{S}}\right)dt^2-\dfrac{4aMr\sin^2\theta}{\mathscr{S}}dtd\phi+\frac{\mathscr{S}}{\mathscr{D}} dr^2  +\mathscr{S} d\theta^2+\dfrac{\mathscr{A}\sin^2\theta}{\mathscr{S}}d\phi^2,
\end{equation}
where $\mathscr{D}\equiv r^2-2Mr+a^2$, $ \mathscr{S}\equiv r^2+a^2 \cos^2\theta $ and $ \mathscr{A}\equiv (r^2 + a^2)^2 -\mathscr{D}a^2 \sin^2\theta$, representing their usual meanings \cite{CH}. We now adopt the Kinnersley null tetrad with the normalisation conditions $k^al_a=-1$ and $m^a\bar{m}_a=1$ \cite{WK}:
\begin{align}
k^a=&\frac{1}{\mathscr{D}}(r^2+{a}^2, \mathscr{D}, 0, {a} ),\nonumber\\
 l^a=& \frac{1}{2\mathscr{S}}(r^2+ {a}^2, -\mathscr{D}, 0, {a} ),\nonumber\\
 m^a=& \frac{1}{\sqrt{2}(r+\iu {a}\cos\theta)}\left(\iu {a}\sin\theta, 0, 1, \frac{\iu}{\sin\theta}\right).
\end{align}
We know that the Kerr metric belongs to the Petrov type D class, and therefore, according to the Goldberg-Sachs theorem \cite{GS}, the principal null congruences are geodesic and shear-free. Hence, we obtain the corresponding spin coefficients
\begin{align}
&\epsilon=\kappa=\lambda=\nu=\sigma=0,\,\,\tau=-\dfrac{\iu a \sin\theta}{\sqrt{2}\mathscr{S}}, \nonumber\\
&\uppi=\dfrac{\iu a \sin\theta}{\sqrt{2}(r-\iu a \cos\theta)^2},\,\,\beta=-\dfrac{\bar{\varrho}\cot\theta}{2\sqrt{2}},\,\,\alpha=\uppi-\bar{\beta},\nonumber\\
&\varrho=-\frac{1}{r-\iu a \cos\theta},\,\,\mu=\frac{\varrho \mathscr{D}}{2\mathscr{S}},\,\,\gamma=\mu+\frac{r-M}{2\mathscr{S}}.
\end{align}
We can now easily isolate the real part of the spin coefficient $\gamma$ and obtain the local $T_{grav}$ (using \eqref{T_new}) in any point of Kerr spacetime as
\begin{equation}\label{T_kerr}
T_{grav}= \frac{1}{\sqrt{2}\pi\mathscr{S}^2}\left[\mathscr{S}(r-M)-\mathscr{D}r\right].
\end{equation}
We observe that the above expression \eqref{T_kerr} also corresponds to the surface gravity expression on the horizon ($\mathscr{D}=0$, i.e., $\mu=0$) for Kerr black hole, derived by Nielsen \cite{nielsen} 
\begin{align}\label{Kkerr}
\mathcal{K}_{\mathcal{H}}\equiv\vert \gamma+\bar{\gamma} \vert=\dfrac{\sqrt{M^2-{a}^2}}{2M\left(M+\sqrt{M^2-{a}^2}\right)-{a}^2\sin^2\theta}.
\end{align}
In the above expression \eqref{Kkerr}, if we put $a=0$, we get back the expression \eqref{Kschw} for the Schwarzschild BH. A similar analysis can be done for the Killing and zero angular momentum observer (ZAMO) frame using the results obtained in \cite{CH} to obtain the corresponding $T_{grav}$. Therefore, the above analysis seems to indicate that the proposal is also valid in a general Petrov type D spacetime, at least in vacuum. 

\section{Physical interpretation of the non-affinity of the null congruence}

To understand transparently, the physical and measurable emergence of the non-affinity factor, let us suppose that light rays (electromagnetic waves) are travelling along the outgoing null geodesics of the Newman Penrose tetrad. In an absolute general context, the associated field variables can be written as $F=A(x^a)\exp\{i\psi\}$, where the function $A(x^a)$ is the amplitude and the function $\psi$ is the ``{\em Eikonal}". In the geometric optics limit, where we assume of wavelength of the field going to zero, the eikonal is a large quantity \cite{LL}. We can immediately identify the tangent vectors for the outgoing real null NP vector $k^a$ as the wave four-vector $n^a$ of the travelling electromagnetic waves, obeying the following conditions
\begin{eqnarray}
n^an_a=0\,,\\
\frac{Dn^a}{\partial\nu}=\mathcal{K}n^a\,,
\end{eqnarray}
where $\nu$ is the non-affine curve parameter along the geodesic.  From the second equation above it is quite clear that as the wave travels along the geodesic the change in the wave vector is in the direction of wave vector itself. In other words, the wave vector is getting {\em re-scaled} as the light ray travels and we have $n^a\rightarrow B(x^a)n^a$, to satisfy the first equation above. We know that the wave vector can be obtained from the eikonal function as
\be
n_a=\frac{\partial\psi}{\partial x^a}\,,
\ee
and the normalisation condition gives the {\it Eikonal equation}
\be
g^{ab}\frac{\partial\psi}{\partial x^a}\frac{\partial\psi}{\partial x^b}=0\;.
\ee
Thus, when the light travels from a point $P_0$ to a future point $P_1$ on the null geodesic, the change in the eikonal, due to the non-affine curve parameter will be given by
\be
\psi(P_1)-\psi(P_0)=\int_{P_0}^{P_1}B(x^a)n_a dx^a
\ee
Now, as discussed in \cite{LL}, the above system of equation can be mapped to the Hamilton-Jacobi equations of the material particles, by identifying the eikonal as the action function $S$. This then readily maps the three dimensional spatial part of the wave vector to the 3-momenta ($p_\alpha$), while the temporal part of the wave vector is mapped to the Hamiltonian ($\mathbb{H}$) of the system
\be
p_\alpha=n_\alpha=\frac{\partial\psi}{\partial x^a}\,, \mathbb{H}= \frac{\partial\psi}{\partial x^0}\,.
\ee
For electromagnetic waves, we know that both the Hamiltonian and the 3-momenta depends on the angular frequency $\omega$, and hence any rescaling of these quantities would result in change of the angular frequency, as the light travels on the manifold with non-trivial intrinsic curvature. Hence, the presence of gravitational temperature can be directly related to the gravitational red/blue shift of the light waves as it travels on the curved manifold. 

\section{Discussion}

Our result on the geometrical identification of the temperature function related to a free gravitational field in locally rotationally symmetric spacetimes, hinges on the following important observation. We can always locally describe an open set of the spacetime manifold using a Newman Penrose null tetrad. However, the presence of non-trivial intrinsic curvature in the open set, would make it impossible to find a curve parameter, which is an affine parameter for both the real null vectors of the null tetrad. Hence at least one of these real null vectors has to be non-affinely parametrised, and without any loss of generality we can assume that the outgoing null congruence is the one. We identified the measure of this non-affinity with the local gravitational temperature on this open set. \\

Interestingly, this identification satisfies all the known temperature functions and also a causal transport equation emerges from the Einstein field equations. Although the mathematical proofs for this identification was constrained to LRS-II symmetries, we showed that the final result carries over to more general algebraically special spacetimes like Kerr spacetime, which somewhat points towards the robustness of the result. We also discussed the relation between this non-affinity and the red/blue shift of electromagnetic waves travelling in a curved manifold.

\begin{acknowledgments}
SC thanks University of KwaZulu-Natal (UKZN), South Africa for post-doctoral scholarship. GA, SDM, RG are supported by the National Research Foundation (South Africa). The authors are thankful to the late Prof. Naresh Dadhich for his insightful comments during the initial phase of the work.
\end{acknowledgments}


\begin{thebibliography}{99}
\bibitem{penrose1} Penrose R 1977 Proc. 1st Marcel Grossmann Meeting on General Relativity ed R Ruffini (North-
Holland) p 173
\bibitem{penrose2} Penrose R 1989 Difficulties with Inflationary Cosmology \textit{Ann. New York Acad. Sci.} \textbf{571} 249

\bibitem{CET} Clifton T, Ellis G F R and Tavakol R 2013 A gravitational entropy proposal \textit{Class. Quantum Grav.} \textbf{30} 125009

\bibitem{E7} Maartens R and Bassett B A 1998 Gravito-electromagnetism \textit{Class. Quantum Grav.} \textbf{15} 705

\bibitem{sen1} Senovilla J M M 2000 Super-energy tensors \textit{Class. Quantum Grav.} \textbf{17} 2799


\bibitem{Bekenstein} Bekenstein J D 1973 Black Holes and Entropy \textit{Phys. Rev. D} \textbf{7} 2333

\bibitem{Hawking} Hawking S W 1975 Particle creation by black holes \textit{Commun. Math. Phys.} \textbf{43} 199--220

\bibitem{AEGH} Acquaviva G, Ellis G F R, Goswami R and Hamid A I M 2015 Constructing black hole entropy from gravitational collapse \textit{Phys. Rev. D} \textbf{91} 064017

\bibitem{SG} Guha S, Khan S and Goswami R 2025 From a collapsing radiating star to an evaporating black hole: A smooth transition from classical to quantum entropy 
\textit{Class. Quantum Grav.} \textbf{43} 07LT02

\bibitem{bolejko}Bolejko K 2018 Gravitational entropy and the cosmological no-hair conjecture \textit{Phys. Rev. D} \textbf{97} 083515

\bibitem{CET1} Sussman R A and Larena J 2014 Gravitational entropies in LTB dust models \textit{Class. Quantum Grav.} \textbf{31} 075021

\bibitem{CET2} Sussman R A and Larena J 2015 Gravitational entropy of local cosmic voids \textit{Class. Quantum Grav.} \textbf{32} 165012

\bibitem{CET3} Piza\~{n}a F A, Sussman R A and Hidalgo J C 2022 Gravitational entropy in Szekeres class I models \textit{Class. Quantum Grav.} \textbf{39} 185005

\bibitem{CET4} Chakraborty S, Guha S and Goswami R 2021 An investigation on gravitational entropy of cosmological models \textit{Int. J. Mod. Phys. D} \textbf{30} 2150051

\bibitem{CET5} Chakraborty S, Guha S and Goswami R 2022 How appropriate are the gravitational entropy proposals for traversable wormholes? \textit{Gen. Relativ. Gravit.} \textbf{54} 47

\bibitem{CET6}Sussman R A, N\'{a}jera S, Piza\~{n}a F A and Hidalgo J C 2025 Gravitational entropy of fluids with energy flux \textit{arXiv preprint} arXiv:2509.25007 [gr-qc]

\bibitem{arrow} Sussman R A, N\'{a}jera S, Piza\~{n}a F A and Hidalgo J C 2026 Revisiting the gravitational `arrow of time' \textit{Class. Quantum Grav.} \textbf{43} 115019

\bibitem{CET7} Li\~{n}an A A P, P\'{e}rez D, Luna C and Romero G E 2025 Gravitational entropy in black hole transformations \textit{Eur. Phys. J. C} \textbf{85} 451


\bibitem{H8} Ellis G F R 1967 The dynamics of pressure-free matter in general relativity \textit{J. Math. Phys.} \textbf{8} 1171--94;\\
van Elst H and Ellis G F R 1996 \textit{Class. Quantum Grav.} \textbf{13} 1099

\bibitem{H9} Ellis G F R and van Elst H 1999 Cosmological models (Carg\`{e}se Lectures 1998) in \textit{Theoretical and Observational Cosmology} ed M Lachieze-Rey (Dordrecht: Kluwer) p 1

\bibitem{H10} Ellis G F R, Maartens R and MacCallum M A H 2007 \textit{Relativistic Cosmology} (Cambridge: Cambridge University Press)

\bibitem{H11} Clarkson C A and Barrett R K 2003 Covariant perturbations of Schwarzschild black holes \textit{Class. Quantum Grav.} \textbf{20} 3855

\bibitem{H12} Betschart G and Clarkson C A 2004 Scalar and electromagnetic perturbations on LRS class II space-times \textit{Class. Quantum Grav.} \textbf{21} 5587

\bibitem{H13} Clarkson C 2007 A Covariant approach for perturbations of rotationally symmetric spacetimes \textit{Phys. Rev. D} \textbf{76} 104034

\bibitem{H14} Stewart J M and Ellis G F R 1968 On solutions of Einstein's equations for a fluid which exhibit local rotational symmetry \textit{J. Math. Phys.} \textbf{9} 1072

\bibitem{R18} Goswami R and Ellis G F R 2021 Tidal forces are gravitational waves \textit{Class. Quantum Grav.} \textbf{38} 085023



\bibitem{E8} Misner C W and Sharp D H 1964 Relativistic equations for adiabatic, spherically symmetric gravitational collapse \textit{Phys. Rev.} \textbf{136} B571

\bibitem{E9} Stephani H, Herlt E, MacCallum M, Hoenselaers C and Kramer D 2003 \textit{Exact Solutions of Einstein's Field Equations} (Cambridge: Cambridge University Press)

\bibitem{chandra} Chandrasekhar S 1983 \textit{The Mathematical Theory of Black Holes} (Oxford: Oxford University Press)

\bibitem{fronov} Frolov V P and Novikov I D 1998 \textit{Black Hole Physics: Basic Concepts and New Developments} (Kluwer Academic, Dordrecht)

\bibitem{null} Goswami R and Ellis G F R 2017 4-dimensional spacetimes from 2-dimensional conformal null data \textit{Class. Quantum Grav.} \textbf{34} 055009

\bibitem{R17} Goswami R and Ellis G F R 2018 Transferring energy in general relativity \textit{Class. Quantum Grav.} \textbf{35} 165007

\bibitem{wr} Wainwright J 1970 A Class of Algebraically Special Perfect Fluid Space-Times \textit{Commun. Math. Phys.} \textbf{17} 42--60

\bibitem{dust} Ellis G F R 1967 Dynamics of pressure-free matter in general relativity \textit{J. Math. Phys.} \textbf{8} 1171
\bibitem{nielsen} Nielsen A 2007 \textit{Black Hole Horizons and Black Hole Thermodynamics} PhD Thesis (University of Canterbury. Physics and Astronomy)

\bibitem{MOTS1} Hamid A I M, Goswami R and Maharaj S D 2014 Cosmic
censorship conjecture revisited: covariantly \textit{Class. Quantum Grav.} \textbf{31} 135010


\bibitem{CH} Hansraj C, Goswami R and Maharaj S D 2023 A semi-tetrad decomposition of the Kerr spacetime \textit{Eur. Phys. J. C} \textbf{83} 321

\bibitem{WK} Kinnersley W 1969 Type D Vacuum Metrics \textit{J. Math. Phys.} \textbf{10} 1195-1203

\bibitem{GS} Goldberg J N and Sachs R K 2009 A theorem on Petrov types \textit{Gen. Relativ. Gravit.} \textbf{41} 433-444

\bibitem{LL} L D Landau and E M Lifshitz (1951), \textit{The Classical Theory of Fields} (Pergamon Press), Chapter 7, Section 53.

\end{thebibliography}
\end{document}